\documentclass[a4paper,11pt]{article}
 \usepackage[english]{babel}
\usepackage[a4paper]{geometry}
\usepackage{amsmath}
\usepackage{amsthm}
\usepackage{amssymb}
\usepackage{mathtools}
\usepackage{bbm}
\usepackage{color}
\usepackage{comment}
\usepackage{mathrsfs}
\usepackage{enumerate}
\usepackage{bm}
\usepackage{sidecap} 

\usepackage{hyperref}         

\numberwithin{equation}{section}

\def\bE{\mathbb{E}}
\def\bP{\mathbb{P}}

\def\cG{\mathcal{G}}

\def\cN{\mathcal{N}}
\def\cE{\mathcal{E}}

\def\eps{\varepsilon}

\DeclareMathOperator{\sech}{sech}

\def\tbf{\textbf }

\def\wh{\widehat}

\def\pf{p_{\text{free}} }

\def\bE{{\mathbb{E}}}

\newtheorem{theorem}{Theorem}[section]  

\newtheorem{prop}[theorem]{Proposition}
\newtheorem{lemma}[theorem]{Lemma}

                          \let\r=r

\begin{document}

\title{The Replica Symmetric Formula for the SK Model Revisited}

\author{Christian Brennecke$^1$\footnote{Corresponding author: brennecke@iam.uni-bonn.de}\;, 
Horng-Tzer Yau$^2$
\\
\\
Institute for Applied Mathematics, University of Bonn, \\
Endenicher Allee 60, 53115 Bonn, Germany$^1$ \\
\\
Department of Mathematics, Harvard University, \\
One Oxford Street, Cambridge MA 02138, USA$^2$ \\
\\}

\maketitle

\begin{abstract}
We provide a simple extension of Bolthausen's Morita type proof of the replica symmetric formula [E. Bolthausen, Stat. Mech. of Classical and Disordered Systems, pp. 63-93 (2018)] for the Sherrington-Kirkpatrick model and prove the replica symmetry for all $(\beta,h)$ that satisfy $\beta^2  E \sech^2(\beta\sqrt{q}Z+h) \leq 1$, where $q = E\tanh^2(\beta\sqrt{q}Z+h)$. Compared to [E. Bolthausen, Stat. Mech. of Classical and Disordered Systems, pp. 63-93 (2018)], the key of the argument is to apply the conditional second moment method to a suitably reduced partition function. 
\end{abstract}

\section{Introduction}

We study systems of $N$ spins $ \sigma_i$, $i\in \{1,\dots,N\}$, with values in $\{-1,1\}$ and with the Hamiltonian $H_N:\{-1,1\}^N\to\mathbb{R}$ defined by
		\begin{equation} \label{eq:HN}
		H_N(\sigma) = \frac{\beta}{\sqrt 2}\sum_{1\leq i, j\leq N} g_{ij} \sigma_{i}\sigma_j + h\sum_{i=1}^N  \sigma_i.
		\end{equation}
The interactions $ \{ g_{ij} \}$ are i.i.d. centered Gaussians of variance $1/N$ for $i\neq j$ and we set $g_{ii}\equiv 0$. $\beta \geq 0$ denotes the inverse temperature and $h > 0 $ the external field strength. 

Eq. \eqref{eq:HN} corresponds to the Sherrington-Kirkpatrick (SK) spin glass model \cite{SK} and we are interested in its free energy $f_N$ at high temperature, where
		\begin{equation} f_N = \frac1N \log Z_N, \,\,\text{  } \,\,Z_N =  \sum_{\sigma\in\{-1,1\}^N} e^{H_N(\sigma)}. \end{equation}
The mathematical understanding of the SK model has required substantial efforts until the famous Parisi formula \cite{Par1, Par2} was rigorously established by Guerra \cite{Gue} and Talagrand \cite{Tal}. Later, Panchenko \cite{Pan1, Pan3} gave another proof based on the ultrametricity of generic models. For a thorough introduction to the topic, we refer to \cite{MPV, Tal1, Tal2, Pan2}. 

Despite the validity of the Parisi formula, it is an interesting question to prove the replica symmetry of the SK model at high temperature, as predicted by de Almeida and Thouless \cite{AT}. Replica symmetry is expected for all $(\beta,h)$ that satisfy
		\begin{equation}\label{eq:ATline}   \beta^2 E  \sech^4(\beta\sqrt{q}Z+h) < 1, \end{equation}
where $q$ denotes the unique solution of the self-consistent equation 
		\begin{equation}\label{eq:q}q =  E \tanh^2(\beta \sqrt{q} Z +h).\end{equation}
In both cases $Z \sim\mathcal{N}(0,1)$ denotes a standard Gaussian and $E$ the expectation over $Z$. 

In the special case that the external field in direction of $\sigma_i$ is a centered Gaussian random variable, $h_i = hg_i$ for i.i.d. $g_i \sim \mathcal{N}(0,1)$ (independent of the $g_{ij}$), replica symmetry has been recently shown in \cite{Ch} for all $(\beta,h)$ that satisfy \eqref{eq:ATline} (in which case $h$ in \eqref{eq:ATline} and \eqref{eq:q} is replaced by $h Z'$ for some $ Z'\sim \cN(0,1)$ independent of $Z$); see also \cite{Sh} for previous results in this case. For Hamiltonians $H_N$ as in \eqref{eq:HN} (or, more generally, Hamiltonians with non-centered random external field), however, replica symmetry is to date only known above the AT line up to a bounded region in the $(\beta,h)$-phase diagram. This has been analyzed in \cite{JagTob}. Like \cite{Ch}, this analysis is based on the Parisi variational problem and we refer to \cite{JagTob} for the details. For previously obtained results based on the Parisi formula, see also \cite[Chapter 13]{Tal2}.

In this note, instead of analyzing the high temperature regime in view of the Parisi variational problem, we give a simple extension of Bolthausen's argument \cite{Bolt2} and prove the replica symmetric formula for all $(\beta,h)$ that satisfy
		 \begin{equation}\label{eq:techline} \beta^2  E  \sech^2(\beta\sqrt{q}Z+h)  \leq 1. \end{equation}
Although \eqref{eq:techline} is clearly stronger than the condition \eqref{eq:ATline}, it already covers a fairly large region of the high temperature regime, see Fig. \ref{fig} for a schematic. It improves upon the inverse temperature range from \cite{Bolt2}, where $\beta $ was assumed to be sufficiently small. 

\begin{theorem}\label{thm:main}
Assume that $(\beta,h)$ satisfies \eqref{eq:techline}, then	
		\begin{equation}\label{eq:RS} \lim_{N\to\infty}\mathbb{E} \frac1N\log Z_N = \log2+  E \log \cosh(\beta\sqrt{q}Z+h) + \frac{\beta^2}4(1-q)^2. \end{equation}
\end{theorem}

\noindent \emph{Remarks:} 
\begin{itemize}
\item[1)] From \cite{GueTo}, it is well-known that $\lim_{N\to\infty}\mathbb{E} \frac1N\log Z_N$ exists and that almost surely $\lim_{N\to\infty} \frac1N \log Z_N = \lim_{N\to\infty}\mathbb{E} \frac1N\log Z_N$.

\item[2)] It follows from the results of \cite{Gue} that the right hand side in \eqref{eq:RS} provides an upper bound to the free energy $\lim_{N\to\infty}\mathbb{E} \frac1N\log Z_N$, for all inverse temperatures and external fields $(\beta,h)$. To establish Theorem \ref{thm:main}, it is therefore sufficient to prove that the right hand side in \eqref{eq:RS} provides a lower bound to $\lim_{N\to\infty}\mathbb{E} \frac1N\log Z_N$. 
\end{itemize}

We conclude this introduction with a quick heuristic outline of the main argument. To this end, consider first the case $h=0$ where the critical temperature corresponds to $\beta =1$. In this case, it is straight forward (see e.g. \cite[Chapter 1, Section 3]{BoBo}) to see that 
		\[ \lim_{N\to\infty} \frac1N\log \mathbb{E} Z_N =  \log 2 +  \frac{\beta^2}4, \hspace{0.5cm} \lim_{N\to\infty} \frac1N\log \mathbb{E} Z_N^2 = 2\log 2 +  \frac{\beta^2}2\]
for all $\beta  <1 $. The replica symmetric formula thus follows from the second moment method, using the Gaussian concentration of the free energy. In fact, for $h=0$, also the fluctuations of $Z_N/\mathbb{E} Z_N$ have been known for a long time \cite{ALR}.

Clearly, it would be desirable to extend this simple argument to the case $h>0$, but a direct application of the second moment method does not work here. However, as suggested in \cite{Bolt2}, one may hope to obtain a model similar to the case $h=0$ by centering the spins around suitable magnetizations and viewing $Z_N$, up to normalization, as an average over the corresponding coin-tossing measure. To center the spins correctly, recall that at high temperature one expects the TAP equations \cite{TAP} to hold, that is
		\begin{equation}\label{eq:TAPeq} m_i \approx \tanh\Big( h + \beta\sum_{j\neq i} \bar{g}_{ij} m_j - \beta^2(1-q)m_i \Big)\hspace{0.5cm}(\text{for }i=1,\dots,N),\end{equation}
where $\bar{g}_{ij} = (g_{ij}+g_{ji})/\sqrt2$ and $m_i =  Z_N^{-1}\sum_{\sigma} \sigma_i \,e^{H_N(\sigma)}$. The validity of \eqref{eq:TAPeq} is known for sufficiently small $\beta$ (see \cite{Tal1, Cha} and, more recently, \cite{ABSY}; see also \cite{AuJag} on the TAP equations for generic models, valid at all temperatures) and expected to be true under \eqref{eq:ATline}. In \cite{Bolt1,Bolt2}, Bolthausen has provided an iterative construction $(\textbf{m}^{(s)})_{s \in\mathbb{N}}$ of the solution to \eqref{eq:TAPeq} that converges (in a suitable sense) in the full high temperature regime \eqref{eq:ATline}. The main result of \cite{Bolt2} is a novel proof of \eqref{eq:RS} for $\beta$ small enough, based on a conditional second moment argument, given the approximate solutions $(\textbf{m}^{(s)})_{s\in\mathbb{N}}$. It has remained an open question, however, if the approach can be extended to the region \eqref{eq:ATline}. 

In this note, while we are not able to resolve this question for all $(\beta,h)$ satisfying \eqref{eq:ATline}, we improve the range of $(\beta, h)$ to \eqref{eq:techline} as follows. \cite{Bolt1,Bolt2} show, roughly speaking, that $\textbf{m}^{(k+1)} \approx \sum_{s=1}^k \gamma_s \phi^{(s)}$ for certain orthonormal vectors $\phi^{(s)} \in \mathbb{R}^N$ and deterministic numbers $\gamma_s \, (\,\approx  \langle \textbf{m}^{(k+1)}, \phi^{(s)} \rangle$ with high probability), where $\langle \tbf{x}, \tbf{y} \rangle = N^{-1}\sum_{i=1}^N x_iy_i$ for $\tbf{x}, \tbf{y}\in \mathbb{R}^N$. One also has $\textbf{g} = \textbf{g}^{(k+1)} + \sum_{s=1}^k\rho^{(s)}$ for the interaction $\textbf{g} = (g_{ij})_{1\leq i,j\leq N}$, where the $\rho^{(s)} \in \mathbb{R}^{N\times N}$ are measurable w.r.t. $(\textbf{m}^{(s)})_{s\leq k+1}$ and where the modified interaction $\textbf{g}^{(k+1)}$ is Gaussian, conditionally on the $(\textbf{m}^{(s)})_{s\leq k+1}$, with the property that $\textbf{g}^{(k+1)} \sum_{s=1}^k \gamma_s \phi^{(s)} = 0$. Up to negligible errors, one obtains with $\wh \sigma = \sigma - \tbf{m}^{(k+1)}$ that
		\begin{equation} \label{eq:mainidea}\begin{split} \frac1N \log Z_N \approx &\,\log 2 +  E \log \cosh(\beta \sqrt{q}Z+h) \\
		& + \frac1N \log \sum_{\sigma\in\{-1,1\}^{N}} \pf(\sigma) \exp\bigg[ \frac{N\beta}{\sqrt{2}}\big \langle\wh \sigma, \textbf{g}^{(k)} \wh \sigma \big\rangle + N \mathcal{O}\big(\max_{s} |  \gamma_s - \langle\sigma, \phi^{(s)}\rangle|\big) \bigg], \end{split}\end{equation}
where $\pf$ denotes the product measure for which $\sum_{\sigma } \pf(\sigma) \,\sigma=\textbf{m}^{(k+1)}$. A simple observation is now that we can ignore the error $N\mathcal{O}(\max_{s} |  \gamma_s - \langle\sigma, \phi^{(s)}\rangle|)$ in \eqref{eq:mainidea} by restricting the modified partition function to those $\sigma$ with $ \max_{s} |  \gamma_s - \langle\sigma, \phi^{(s)}\rangle| \approx 0$. Note that the probability of the complement of this set is small under $\pf$, because $\gamma_s\approx  \langle \textbf{m}^{(k+1)}, \phi^{(s)} \rangle$. This yields a simple lower bound on $\frac1N\log Z_N$ and we can apply the conditional second moment argument to the restricted partition function. We show that its first conditional moment equals $ \beta^2 (1-q)^2/4$ (up to negligible errors) in the full high temperature regime \eqref{eq:ATline}. To dominate its second moment by the square of the first, on the other hand, we need to impose the stronger condition \eqref{eq:techline}. 

Notice that imposing similar orthogonality restrictions on the partition function has been proved useful before for obtaining lower bounds on the free energy, like in the TAP analysis of the spherical SK model \cite{BelKis}. 

\begin{figure}\label{fig}
\centering
\includegraphics[width=0.52\textwidth]{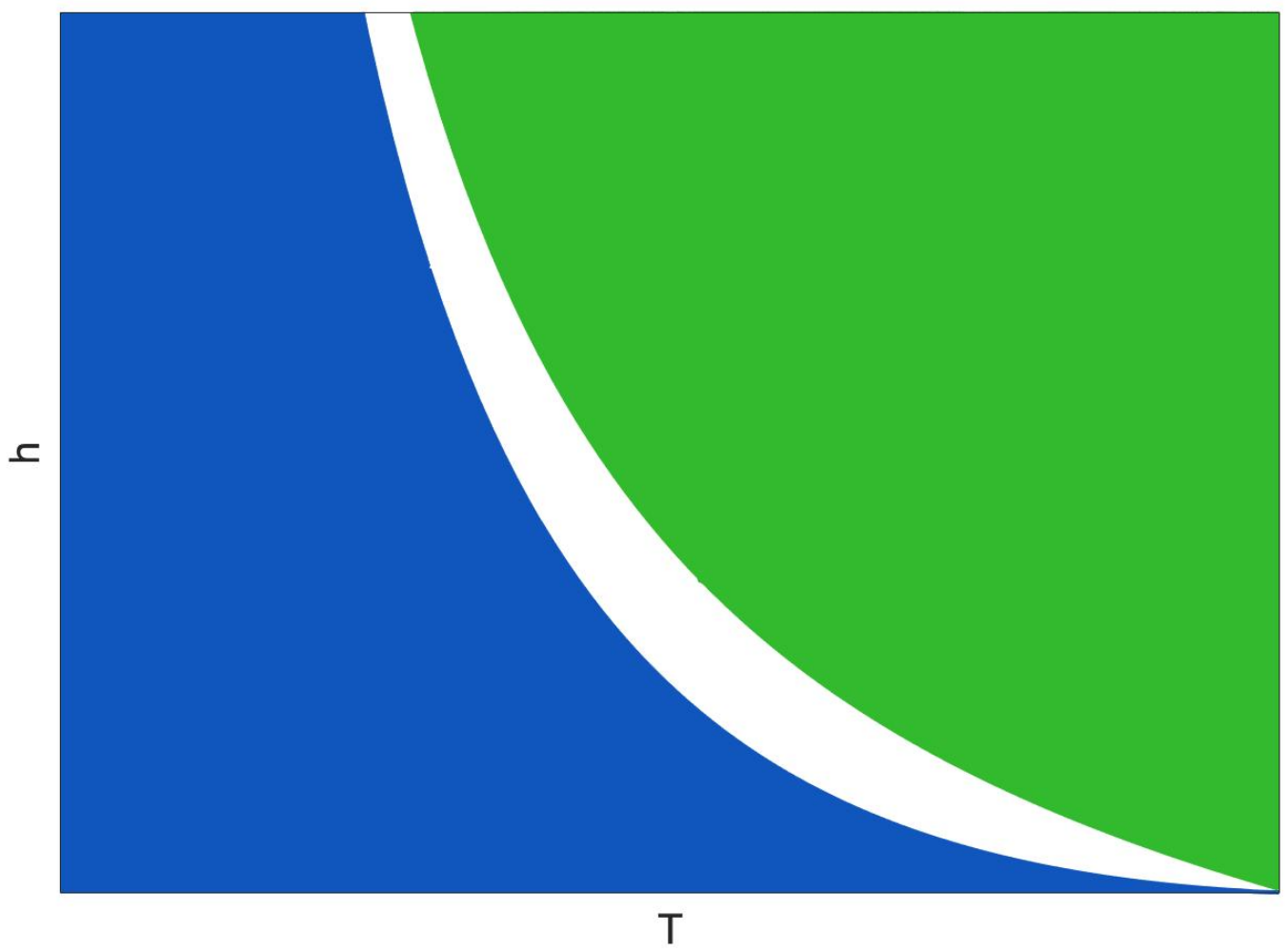}
\caption{Schematic of the $(T,h)$ phase diagram, where $T=\frac1\beta$ denotes the temperature. In the blue region, whose boundary corresponds to the AT line \eqref{eq:ATline}, the SK model is known to be replica symmetry breaking.  The boundary of the green region corresponds to the condition \eqref{eq:techline}. Theorem \ref{thm:main} proves the replica symmetry in the green region.}
\end{figure}

Although \eqref{eq:techline} covers already a comparably large region of the high temperature phase as schematically shown in Fig. \ref{fig}, it remains an open question whether the second moment argument can be extended to the full high temperature regime \eqref{eq:ATline}, see also related comments in \cite[Section 6]{Bolt2}.


The paper is structured as follows. In the next two sections we set up the notation and recall Bolthausen's iterative construction of the magnetization \cite{Bolt1, Bolt2}. In Section \ref{sec:2ndmoment}, we define the reduced partition function and compute its first and second moments. In the last section, we apply the conditional second moment method to prove Theorem \ref{thm:main}.

\section{Notation}\label{sec:notation}

In this section, we introduce basic notation and conventions. We follow closely \cite{Bolt2}. 

We usually denote vectors in $\mathbb{R}^N$ by boldface or greek letters. If $\tbf{x}\in\mathbb{R}^N$, and $g:\mathbb{R}\to\mathbb{R}$, we define $g(\tbf{x})$ in the componentwise sense. By $\langle\cdot, \cdot\rangle: \mathbb{R}^N\times\mathbb{R}^N\to\mathbb{R}$, we denote the normalized inner product
		\[ \langle \tbf{x}, \tbf{y}\rangle = \frac1N \sum_{i=1}^N x_i y_i\]
and by $\| \cdot\| =\sqrt{\langle\cdot, \cdot\rangle}$ the induced norm. 
We also normalize the tensor product $\tbf{x}\otimes \tbf{y}: \mathbb{R}^N\to \mathbb{R}^N$ of two vectors $\tbf{x}, \tbf{y}\in\mathbb{R}^N$ so that for all $\tbf{z}\in\mathbb{R}^N$
		\[ (\tbf{x}\otimes \tbf{y}) (\tbf{z}) =  \langle \tbf{y},\tbf{z}\rangle\, \tbf{x}.\]

Given a matrix $\tbf A\in \mathbb{R}^{N\times N}$, ${\tbf{A}}^T\in \mathbb{R}^{N\times N}$ denotes its transpose and $ \bar{\tbf{A}} \in\mathbb{R}^{N\times N}$ denotes its symmetrization
		\[  \bar{\tbf{A}}  = \frac{\tbf{A}+\tbf{A}^T}{\sqrt{2}}.\]

We mostly use the letters $Z, Z', Z_1, Z_2,$ etc. to denote standard Gaussian random variables independent of the disorder $\{g_{ij}\}$ and independent from one another. When we average over such Gaussians, we denote the corresponding expectation by $E$, to distinguish it from the expectation $\mathbb{E}$ with respect to the disorder $\{g_{ij}\}$. Unless specified otherwise, we consider all Gaussians to be centered.

Finally, given two sequences of random variables $ (X_N)_{N\geq 1}, (Y_N)_{N\geq 1}$ that may depend on parameters like $\beta, h,$ etc. we say that 
		\[ X_N\simeq Y_N\]
if and only if there exist positive constants $c, C >0$, which may depend on the parameters, but which are independent of $N $, such that for every $t >0$ we have
		\[ \mathbb{P} ( |X_N-Y_N| > t ) \leq C e^{-c t^2N}.\] 

\section{Bolthausen's Construction of the Local Magnetizations}\label{sec:mi}

In this section, we recall Bolthausen's iterative construction of the solution to the TAP equations \cite{Bolt1,Bolt2} and list the properties that we will need for the proof of Theorem \ref{thm:main}. We follow here the conventions of \cite{Bolt2} and we refer to \cite[Sections 2, 4 \& 5]{Bolt2} for the proofs of the following statements. 

First of all, we define three sequences $(\alpha_{k})_{k\in\mathbb{N}}$, $(\gamma_k)_{k\in\mathbb{N}}$ and $(\Gamma_k)_{k\in\mathbb{N}}$. Set
		\[ \alpha_1 = \sqrt{q}\gamma_1,\hspace{0.5cm} \gamma_1 = E \tanh(\beta\sqrt q Z +h ), \hspace{0.5cm} \Gamma_1^2 = \gamma_1^2,\]
where here and in the following $q$ denotes the unique solution of \eqref{eq:q}. Then, we define $\psi: [0,q]\to [0,q]$ by
		\[ \psi(t)=E \tanh(\beta \sqrt t Z +\beta \sqrt{q-t}Z' +h ) \tanh (\beta \sqrt t Z +\beta \sqrt{q-t}Z'' +h ) \]
and set recursively
		\[\alpha_{k} =  \psi(\alpha_{k-1}), \hspace{0.5cm} \gamma_k = \frac{\alpha_{k} - \Gamma^2_{k-1}}{\sqrt{q - \Gamma^2_{k-1}}}, \hspace{0.5cm} \Gamma_k^2 = \sum_{j=1}^k \gamma_k^2. \]
The following lemma collects important properties of $(\alpha_{k})_{k\in\mathbb{N}}$ , $(\gamma_k)_{k\in\mathbb{N}}$ and $(\Gamma_k)_{k\in\mathbb{N}}$.
\begin{lemma}{(\cite[Lemma 2.2, Corollary 2.3, Lemma 2.4]{Bolt1}, \cite[Lemma 2]{Bolt2})}\label{lm:seqlemma}
\begin{enumerate}[1)]
\item $\psi$ is strictly increasing and convex in $[0,q]$ with $0< \psi(0) <\psi(q) = q$. If \eqref{eq:ATline} is satisfied, then $q$ is the unique fixed point of $\psi$ in $[0,q]$.
\item The sequence $(\alpha_{k})_{k\in\mathbb{N}}$ is increasing and $\alpha_{k} >0$ for every $k\in\mathbb{N}$. If \eqref{eq:ATline} is satisfied, then $\lim_{k\to\infty} \alpha_{k} =q$ and if \eqref{eq:ATline} is satisfied with a strict inequality, the convergence is exponentially fast. 
\item For all $k \geq 2$, we have that $ 0< \Gamma_{k-1}^2<\alpha_{k} <q $ and that $0<\gamma_k < \sqrt{q -\Gamma_{k-1}^2}$. If \eqref{eq:ATline} is true, then $\lim_{k\to\infty} \Gamma_k^2 =q$ and, as a consequence, $\lim_{k\to\infty} \gamma_k =0$. 
\end{enumerate}
\end{lemma}

Next, we recall Bolthausen's modified interaction matrix. We define 
		\[ \tbf{g}^{(1)} = \tbf g, \hspace{0.5cm} \phi^{(1)} = \tbf{1} \in \mathbb{R}^N, \hspace{0.5cm} \tbf{m}^{(1)}=\sqrt{q}\tbf{1}\in\mathbb{R}^N. \]
Assuming $\tbf{g}^{(s)}, \phi^{(s)}, \tbf{m}^{(s)}$ are defined for $1\leq s\leq k$, we set 
		\[ \zeta^{(s)} = \bar{\tbf g}^{(s)}\phi^{(s)}\]
and we define the $\sigma$-algebra $\cG_k$ through
		\[ \cG_k = \sigma\big( \tbf{g}^{(s)} \phi^{(s)},(\tbf{g}^{(s)})^T \phi^{(s)}: 1\leq s\leq k   \big). \]
Expectations with respect to $\cG_k$ are denoted by $\bE_k$. Furthermore, we set
		\begin{equation}\begin{split}\label{eq:mk=1}
		\tbf{h}^{(k+1)} &= h\tbf{1} + \beta \sum_{s=1}^{k-1}\gamma_s\zeta^{(s)} + \beta \sqrt{q-\Gamma_{k-1}^2}\zeta^{(k)}, \hspace{0.5cm} \tbf{m}^{(k+1)} =\tanh (\tbf{h}^{(k+1)}),\\
		\phi^{(k+1)} & = \frac{\tbf{m}^{(k+1)} - \sum_{s=1}^{k}\langle \tbf{m}^{(k+1)}, \phi^{(s)}\rangle \phi^{(s)}}{\Big\| \tbf{m}^{(k+1)} - \sum_{s=1}^{k}\langle \tbf{m}^{(k+1)}, \phi^{(s)}\rangle \phi^{(s)} \Big\|}
		\end{split}\end{equation}
and we note that $\phi^{(k+1)}$ is $\mathbb{P}-a.s.$ well-defined for all $k$ if $k<N$ \cite[Lemma 5]{Bolt2}. Finally, the modified interaction matrix $\tbf{g}^{(k+1)}$ is defined by
		\[\begin{split}
		\tbf{g}^{(k+1)} =&\, \tbf{g}^{(k)}- \rho^{(k)}, \hspace{0.5cm}\text{where}  \\
		 \rho^{(k)} =&\,  \tbf{g}^{(k)}\phi^{(k)}\otimes \phi^{(k)} + \phi^{(k)}\otimes (\tbf{g}^{(k)})^T\phi^{(k)} -\langle \tbf{g}^{(k)}\phi^{(s)} , \phi^{(k)}\rangle\, \phi^{(k)}\otimes \phi^{(k)}.
		\end{split}\]
In particular, this means that $\bar{\tbf{g}}^{(k+1)}$ is equal to
		\[\begin{split}
		\bar{\tbf{g}}^{(k+1)} = \bar{\tbf{g}}^{(k)}-\bar{\rho}^{(k)},   \hspace{0.5cm}\bar{\rho}^{(k)} = \zeta^{(k)}\otimes \phi^{(k)} + \phi^{(k)}\otimes \zeta^{(k)} -\langle \zeta^{(k)}, \phi^{(k)}\rangle \,\phi^{(k)}\otimes \phi^{(k)}.
		\end{split}\]
It is clear that $(\phi^{(s)})_{s=1}^k$ forms an orthonormal sequence of vectors in $\mathbb{R}^N$ and we denote by $\tbf{P}^{(k)}$ and $\tbf{Q}^{(k)}$ the corresponding orthogonal projections in $\mathbb{R}^N$, that is
		\[ \tbf{P}^{(k)} = \sum_{s=1}^{k} \phi^{(s)}\otimes \phi^{(s)} = (P^{(k)}_{ij})_{1\leq i,j\leq N}, \hspace{0.5cm} \tbf{Q}^{(k)}= \tbf{1}-\tbf{P}^{(k)}= (Q^{(k)}_{ij})_{1\leq i,j\leq N}.  \]
By \cite[Lemma 3]{Bolt2}, $\tbf{m}^{(k)}$ and $\phi^{(k)}$ are $\cG_{k-1}$-measurable for all $k\in \mathbb{N}$ and we also have that
		\[ \tbf{g}^{(k)}\phi^{(s)} =  (\tbf{g}^{(k)})^{T}\phi^{(s)}=\bar{\tbf{g}}^{(k)}\phi^{(s)}  =0, \,\,\forall \,\,s<k.  \]
\begin{prop}{(\cite[Prop. 4]{Bolt2})} \label{prop:intmatrix}
\begin{enumerate}[1)]
\item Conditionally on $\cG_{k-2}$, $ \tbf{g}^{(k)}$ and $ \tbf{g}^{(k-1)}$  are Gaussian with conditional covariance, given $\cG_{k-2}$, equal to
		\[\bE_{k-2} \, g_{ij}^{(k)}g_{st}^{(k)} = \frac1N  Q^{(k-1)}_{is} Q^{(k-1)}_{jt}. \]
\item Conditionally on $\cG_{k-2}$, $\tbf{g}^{(k)}$ is independent of $\cG_{k-1}$. In particular, conditionally on $\cG_{k-1}$, $\tbf{g}^{(k)}$ is Gaussian with the same covariance as in 1).
\item Conditionally on $\cG_{k-1}$, the random variables $\zeta^{(k)}$ are Gaussian with 
		\[ \bE_{k-1} \zeta^{(k)}_i\zeta^{(k)}_j = Q^{(k-1)}_{ij} + \frac1N \phi^{(k)}_i\phi^{(k)}_j.\]
\end{enumerate}
\end{prop}
The main result of \cite{Bolt1} is summarized in the following proposition. 
\begin{prop}{(\cite[Prop. 2.5]{Bolt1}, \cite[Prop. 6]{Bolt2})} \label{prop:mi}
 For every $k\in\mathbb{N}$ and $  s < k$, one has
		\[ \langle \tbf{m}^{(k)}, \phi^{(s)}\rangle \simeq \gamma_s,\hspace{0.5cm}\langle \tbf{m}^{(k)}, \tbf{m}^{(s)}\rangle \simeq \alpha_s, \hspace{0.5cm} \langle \tbf{m}^{(k)},\tbf{m}^{(k)}  \rangle \simeq q.  \]
\end{prop}
The next lemma collects a few auxiliary results that are helpful in the sequel.
\begin{lemma}{(\cite[Lemmas 11, 14, 15(b)]{Bolt2})}\label{lm:aux}
\begin{enumerate}[1)]
\item For every $k\in\mathbb{N}$, $\langle \phi^{(k)}, \zeta^{(k)}\rangle =\sqrt 2\,\langle \phi^{(k)}, \textbf{g}^{(k)}\phi^{(k)}\rangle $ is unconditionally Gaussian with variance $2/N$.
\item For every Lipschitz continuous $f:\mathbb{R}\to\mathbb{R}$ with $|f(x)|\leq C(1+|x|)$ for some $C>0$, one has for all $k\geq 2$ that
		\[\lim_{N\to\infty} \bE \bigg| \frac1N\sum_{i=1}^N f(h_i^{(k+1)}) - Ef(\beta\sqrt{q}Z+h)\bigg|=0.\]
\item For every $k\in \mathbb{N}$ and $t>0$ it holds true that
		\[\lim_{N\to\infty} \bP\big( \big\{ \|\zeta^{(k)}\|\geq 1+t\big\}\big) =0.\]

\end{enumerate}
\end{lemma}


\section{Conditional Moments of Reduced Partition Function}\label{sec:2ndmoment}

Using Bolthausen's magnetizations, we compute in this section the first two conditional moments of a suitably reduced partition function onto a subset $S\subset \{-1,1\}^N$ defined below. This will suffice to establish Theorem \ref{thm:main}, as explained in Section \ref{sec:proofmain}.

Let $\eps>0$ and $ k  \in \mathbb{N} $ be fixed. We define the set $S_{\eps,k} \subset \{-1,1\}^N$ through
		\begin{equation}\label{eq:Sepsk}
		\begin{split}
		 S_{\eps,k} = \Big\{ \sigma\in \{-1,1\}^N: | \langle \sigma- \tbf{m}^{(k+1)}, \phi^{(s)}\rangle | \leq \eps/k,\, \,\,\forall\,\, 1\leq s\leq k\Big\},  
		 \end{split}\end{equation}
with $\phi^{(s)}, \gamma_s$ from Section \ref{sec:mi}. We define the reduced partition function $Z_N^{(k+1)}(S_{\eps,k})$ by
		\begin{equation}\label{eq:ZNred}
		\begin{split} Z^{(k+1)}_N(S_{\eps,k}) &= \sum_{\sigma\in S_{\eps,k} } \pf(\sigma) \exp \bigg( \frac{N\beta}{\sqrt 2} \langle \sigma, \tbf{g}^{(k+1)} \sigma\rangle \bigg)\\
		& = \sum_{\sigma\in S_{\eps,k} } \pf(\sigma) \exp \bigg( \frac{N\beta}2 \langle \sigma, \bar{\tbf{g}}^{(k+1)} \sigma\rangle \bigg), 
		\end{split}\end{equation}
where $\pf:\{-1,1\}^N\to (0,1)$ denotes the coin-tossing measure 
		\begin{equation} \label{eq:pfree} \pf(\sigma) = \prod_{i=1}^N \frac12 \frac{ \exp \big(h_i^{(k+1)}\sigma_i \big) }{ \cosh\big(h_i^{(k+1)}\big)}. \end{equation}
The following lemma records that $ \pf(S_{\eps,k}^c)$ is exponentially small in $N$. 
 \begin{lemma}\label{lm:conc}
 Let $\eps>0$, $k\in \mathbb{N}$, and let $S_{\eps,k} $ and $\pf$ be defined as in \eqref{eq:Sepsk} and \eqref{eq:pfree}, respectively. Then, 
there exist $c, C >0$, independent of $N$ and $\eps$, such that 
		\begin{equation}\label{eq:concpfree}
		 p_{\emph{free}}(S_{\eps,k}) \geq 1-   C e^{- c N\eps^2 }. 
		\end{equation} 
\end{lemma}	
\begin{proof}
By a standard union bound, we have that
		\[\begin{split}
		\pf(S_{\eps,k}^c) 
		\leq &\, k \max_{s=1,\dots,k} \pf \Big(\Big\{ \sigma\in \{-1,1\}^N:  \langle \sigma, \phi^{(s)}\rangle - \langle \tbf{m}^{(k+1)}, \phi^{(s)}\rangle  >  \frac\eps k  \Big\}\Big) \\
		&\, +  k \max_{s=1,\dots,k} \pf \Big(\Big\{ \sigma\in \{-1,1\}^N: \langle \tbf{m}^{(k+1)}, \phi^{(s)}\rangle- \langle \sigma, \phi^{(s)}\rangle  >  \frac\eps k  \Big\}\Big), 
		\end{split} \]
which implies
		\[\begin{split}
		&\pf(S_{\eps,k}^c)\\
		& \leq  k \max_{s=1,\dots,k} \inf_{\lambda \geq 0} \exp \bigg[-N \bigg( \frac{\lambda \eps}k - \frac1N\sum_{i=1}^N  \log\frac{\cosh\big( h_i^{(k+1)}+\lambda \phi_i^{(s)}\big)}{\cosh\big( h_i^{(k+1)}\big)} + \lambda\langle \tbf{m}^{(k+1)}, \phi^{(s)}\rangle\bigg) \bigg]\\
		 &\hspace{0.5cm}+ k \max_{s=1,\dots,k} \inf_{\lambda \geq 0} \exp \bigg[-N \bigg( \frac{\lambda \eps}k - \frac1N\sum_{i=1}^N  \log\frac{\cosh\big( h_i^{(k+1)}-\lambda \phi_i^{(s)}\big)}{\cosh\big( h_i^{(k+1)}\big)} - \lambda\langle \tbf{m}^{(k+1)}, \phi^{(s)}\rangle\bigg) \bigg].\\
		\end{split} \]
Using the pointwise bound $ \log\cosh(x+y)\leq \log\cosh(x) + y\tanh(x) + \frac{y^2}2 $ for $x,y\in\mathbb{R}$, $\langle\phi^{(s)}, \phi^{(s)}\rangle=1$ and the identity $\tanh\big( \tbf{h}^{(k+1)}\big) = \tbf{m}^{(k+1)} $, we obtain
		\[\begin{split}
		\pf(S_{\eps,k}^c) \leq &\, 2k  \inf_{\lambda \geq 0} \exp \bigg[- \frac{N \lambda   \eps}k   + \frac{N\lambda^2}2 \bigg] = 2 k  e^{-   N\eps^2/(2k^2)}.
		\end{split} \]
This concludes \eqref{eq:concpfree} for $c=c_k= 1/(2k^2), C =C_k= 2k$. \end{proof}

We notice that the constants $ c, C>0 $ in \eqref{eq:concpfree} are independent of the realization of the disorder $\{g_{ij}\}$. Thus, $a.s.$ in the disorder (so that $\phi^{(s)}$, $s=1,\dots,k$, and hence $\bar{\tbf{g}}^{(k+1)}$ as well as $\pf$ are well-defined), $S_{\eps,k} \neq \emptyset$ for $N$ large enough. 

The next lemma determines the first conditional moment of the reduced partition function $Z^{(k+1)}_N(S_{\eps,k})$ and is valid in the full high temperature regime \eqref{eq:ATline}. 
\begin{lemma}\label{lm:1stmoment}
Let $\eps>0 $, $k\in \mathbb{N}$ and let $S_{\eps,k}$, $Z^{(k+1)}_N(S_{\eps,k})$ and $\pf$ be as in \eqref{eq:Sepsk}, \eqref{eq:ZNred} and \eqref{eq:pfree}, respectively. Assume that $(\beta,h)$ satisfy the AT condition \eqref{eq:ATline}. Then 
		\begin{equation}\label{eq:1stmoment}
		\lim_{\eps \to 0} \lim_{k\to \infty} \limsup_{N\to\infty} \mathbb{E} \bigg| \frac1N \log \bE_k \, Z_N^{(k+1)}(S_{\eps,k}) - \frac{\beta^2}4(1-q)^2\bigg| = 0
		\end{equation}
\end{lemma}	
\begin{proof}
By Proposition \ref{prop:intmatrix}, we have that 
		\[\begin{split}
		\bE_k Z^{(k+1)}_N(S_{\eps,k}) & = \sum_{\sigma\in S_{\eps,k} } \pf(\sigma) \exp \bigg( \frac{\beta^2N^2}{4} \bE_k \langle \sigma, \tbf{g}^{(k+1)}\sigma\rangle^2 \bigg)\\
		& =  \sum_{\sigma\in S_{\eps,k} } \pf(\sigma) \exp \bigg[ \frac{\beta^2N}{4}  \bigg(  1- \sum_{s=1}^k \langle \sigma,\phi^{(s)}\rangle^2\bigg)^2 \bigg].
		\end{split}\]
Centering around $\tbf{m}^{(k+1)}$ yields
		\[\begin{split}
		1- \sum_{s=1}^k  \langle \sigma,\phi^{(s)}\rangle^2 
		& = 1- \sum_{s=1}^k \langle \tbf{m}^{(k+1)}, \phi^{(s)}\rangle^2 - 2 \sum_{s=1}^k \langle \sigma- \tbf{m}^{(k+1)}, \phi^{(s)}\rangle\langle \tbf{m}^{(k+1)}, \phi^{(s)}\rangle\\
		&\hspace{0.5cm} -\sum_{s=1}^k \langle \sigma - \tbf{m}^{(k+1)}, \phi^{(s)}\rangle^2 
		\end{split}\]	
so that 
		\[\begin{split}
		&\sup_{\sigma \in S_{\eps,k}} \bigg| \bigg(1- \sum_{s=1}^k   \langle \sigma,\phi^{(s)}\rangle^2\bigg) - \bigg(1- \sum_{s=1}^k \langle \tbf{m}^{(k+1)}, \phi^{(s)}\rangle^2\bigg)\bigg|\\
		&\hspace{5cm} \leq  C \sup_{\sigma \in S_{\eps,k}} \sum_{s=1}^k |\langle \sigma - \tbf{m}^{(k+1)}, \phi^{(s)}\rangle|\leq C \eps 
		 \end{split}\]
for some $C>0$, independent of $N$ and $k$. This implies with Lemma \ref{lm:conc} that 
		\[\begin{split} &  \bigg| \frac1N\log \bE_k Z^{(k+1)}_N(S_{\eps,k})  -  \frac{\beta^2}{4} \bigg(1- \sum_{s=1}^k \langle \tbf{m}^{(k+1)}, \phi^{(s)}\rangle^2\bigg)^2  \bigg|   \leq  C \beta^2  \eps +\frac CN | \log( 1- Ce^{-c N\eps^2})|.
		\end{split}\]
Moreover, by Proposition \ref{prop:mi}, we have that	
		$ \lim_{N\to\infty}  \sum_{s=1}^k \langle \tbf{m}^{(k+1)}, \phi^{(s)}\rangle^2  =  \Gamma_{k}^2      $
in $L^p(d\mathbb{P})$, for any $p \in [1;\infty)$, so that 
		\[ \limsup_{N\to\infty} \mathbb{E} \bigg| \frac1N   \log \bE_k Z^{(k+1)}_N(S_{\eps,k}) - \frac{\beta^2}4 \big(1-\Gamma_k^2\big)^2\bigg| \leq C\beta^2\eps .  \]
Finally, since $\lim_{k\to \infty} \Gamma_k^2= q$ under the AT condition \eqref{eq:ATline}, by Lemma \ref{lm:seqlemma}, we let $N\to\infty$, then $k\to\infty $ and then $\eps\to 0$ which implies
		\[\lim_{\eps \to 0} \lim_{k\to \infty} \limsup_{N\to\infty} \mathbb{E} \bigg| \frac1N \log \bE_k \, Z_N^{(k+1)}(S_{\eps,k}) - \frac{\beta^2}4(1-q)^2\bigg| = 0.\]
\end{proof}

The following lemma computes the second conditional moment of $Z_N^{(k+1)}(S_{\eps,k})$ under the stronger high temperature condition \eqref{eq:techline}.
\begin{lemma}\label{lm:2ndmoment}
Let $\eps>0 $, $k\in \mathbb{N}$ and let $S_{\eps,k}$, $Z^{(k+1)}_N(S_{\eps,k})$ and $\pf$ be as in \eqref{eq:Sepsk}, \eqref{eq:ZNred} and \eqref{eq:pfree}, respectively. Assume that $(\beta,h)$ satisfy the condition \eqref{eq:techline}. Then
		\begin{equation}\label{eq:2ndmoment}
		\lim_{\eps \to 0}\lim_{k\to \infty} \limsup_{N\to\infty} \mathbb{E} \bigg| \frac1N \log \bE_k \, \Big[ \big(Z_N^{(k+1)}(S_{\eps,k})\big)^2 \Big] - \frac{\beta^2}2(1-q)^2\bigg| =0.
		\end{equation}
\end{lemma}	
\begin{proof}
Proceeding as in the previous proposition, we compute
		\[\begin{split}
		 & \bE_k \, \Big[ \big(Z_N^{(k+1)}(S_{\eps,k})\big)^2 \Big]\\
		 &= \sum_{\sigma, \tau\in S_{\eps,k} } \pf(\sigma)\pf(\tau) \exp \bigg( \frac{\beta^2N^2}{4} \bE_k \Big(\langle \sigma,  \tbf{g}^{(k+1)}\sigma\rangle+\langle \tau,  \tbf{g}^{(k+1)}\tau\rangle\Big)^2 \bigg)\\
		& =  \sum_{\sigma,\tau \in S_{\eps,k} } \pf(\sigma)\pf(\tau) \exp \bigg[ \frac{\beta^2N}{4}  \bigg(  1- \sum_{s=1}^k \langle \sigma,\phi^{(s)} \rangle^2\bigg)^2+\frac{\beta^2N}{4}  \bigg(  1- \sum_{s=1}^k \langle \tau,\phi^{(s)} \rangle^2\bigg)^2 \bigg]\\
		&\hspace{3.8cm}\times \exp \bigg[ \frac{\beta^2N}{2}  \bigg(  \langle \sigma, \tau\rangle- \sum_{s=1}^k \langle \sigma,\phi^{(s)}\rangle\langle \phi^{(s)} \tau \rangle\bigg)^2 \bigg]\\
		&=  \sum_{\sigma,\tau \in S_{\eps,k} } \pf(\sigma)\pf(\tau) \exp \bigg[ \frac{\beta^2N}{2}    \langle \sigma, \tbf{Q}^{(k)} \tau \rangle ^2 \bigg] \\
		&\hspace{1.5cm}\times \exp\bigg[ \frac{\beta^2N}2  \bigg(  1- \sum_{s=1}^k \langle \textbf{m}^{(k+1)},\phi^{(s)} \rangle^2\bigg)^2 +N\mathcal{O}(\beta^2\eps) \bigg].	
		\end{split}\]
Arguing as in the previous lemma, we therefore see that it is enough to show that 
		\[ \bE\bigg|\frac1N\log \sum_{\sigma,\tau \in S_{\eps,k} } \pf(\sigma)\pf(\tau) \exp \bigg[ \frac{\beta^2N}{2}    \langle \sigma, \tbf{Q}^{(k)} \tau \rangle ^2 \bigg]\bigg|\]
vanishes when $N\to \infty$. To this end, recall that by definition of $S_{\eps,k}$, we have that
		\begin{equation}\label{eq:aux4} \sup_{\tau\in S_{\eps,k}} \big| \|\tbf{Q}^{(k)}\tau\|^2  - \big( 1-\langle \tbf{m}^{(k+1)}, \tbf{P}^{(k)}\tbf{m}^{(k+1)}\rangle\big)\big|\leq  C\eps \end{equation}
for some $C>0$ independent of $N$ and $k$. For fixed $\tau\in S_{\eps,k}$, we then have 
		\[\begin{split}
		 \pf(S_{\eps,k})&\leq \sum_{\sigma \in S_{\eps,k} } \pf(\sigma) \exp \bigg[ \frac{\beta^2N}{2}    \langle \sigma, \tbf{Q}^{(k)} \tau \rangle ^2 \bigg] \\
		& \leq  \pf(S_{\eps,k})+ \int_0^1 dt\, N \beta^2 t  \, e^{  \frac N2 \beta^2  t^2} \, \pf \big( \big\{ \sigma  \in  \{-1,1\}^N: | \langle \sigma, \tbf{Q}^{(k)} \tau \rangle | >t \big\} \big).
		\end{split}\] 
Setting $\lambda = t/(1-q)  $  and using that $\log \cosh(x+y)\leq \log\cosh(x)+y\tanh(x) + y^2/2$ for $x,y\in\mathbb{R}$, we can estimate the tail probability in the integral by  
		\[\begin{split}
		&\pf \big( \big\{ \sigma  \in  \{-1,1\}^N: | \langle \sigma, \tbf{Q}^{(k)} \tau \rangle | >t \big\} \big) \\
		&\leq \exp\bigg[ -N \lambda t  + N  \lambda \langle \tbf{m}^{(k+1)},  \tbf{Q}^{(k)} \tau\rangle +  \frac {N \lambda^2}2 \| \tbf{Q}^{(k)} \tau \|^2 \bigg]  \\
		&\hspace{0.5cm} + \exp\bigg[ -N \lambda t  - N \lambda \langle \tbf{m}^{(k+1)},  \tbf{Q}^{(k)} \tau\rangle +  \frac {N \lambda^2}2 \| \tbf{Q}^{(k)} \tau \|^2 \bigg] \\
		& \leq 2 \exp \bigg[ -\frac{ N t^2}{1-q}   +   \frac{ N t^2}2 \frac{\big\| \tbf{Q}^{(k)} \tau\big\|^2}{(1-q)^2} \bigg] \exp\bigg[\frac{Nt}{1-q} \|\tbf{Q}^{(k)} \tbf{m}^{(k+1)} \|\bigg],
		\end{split}\]		
so that 
		\[\begin{split}
		 \pf(S_{\eps,k})&\leq \sum_{\sigma \in S_{\eps,k} } \pf(\sigma) \exp \bigg[ \frac{\beta^2N}{2}    \langle \sigma, \tbf{Q}^{(k)} \tau \rangle ^2 \bigg] \\
		& \leq \pf(S_{\eps,k})
		+2 \int_0^1dt\,  N \beta^2 t  \,  e^{ \frac{Nt^2}{2(1-q)}\big(\beta^2(1-q)-2 +   \frac{ \| \tbf{Q}^{(k)}\tau\|^2}{(1-q)}\big)}e^{\frac{N t}{1-q}\|\tbf{Q}^{(k)} \tbf{m}^{(k+1)} \|}.
		\end{split}\] 		
In particular, by \eqref{eq:aux4} and because
		\[  \lim_{N\to\infty}  \langle \tbf{m}^{(k+1)}, \tbf{P}^{(k)}\tbf{m}^{(k+1)}\rangle =   \Gamma_k^2, \hspace{0.5cm} \lim_{N\to\infty} \|  \tbf{Q}^{(k)} \tbf{m}^{(k+1)}\|^2 = q-\Gamma_k^2 , \]
in $L^p(d\mathbb{P})$ for $p\in [1;\infty)$, we obtain under the condition \eqref{eq:techline}, i.e. $\beta^2(1-q)\leq1$, that
		\[\lim_{\eps\to 0}\lim_{k\to\infty} \limsup_{N\to\infty} \bE\bigg|\frac1N\log \sum_{\sigma,\tau \in S_{\eps,k} } \pf(\sigma)\pf(\tau) \exp \bigg[ \frac{\beta^2N}{2}    \langle \sigma, \tbf{Q}^{(k)} \tau \rangle ^2 \bigg]\bigg| =0, \]
which implies \eqref{eq:2ndmoment}.
\end{proof}

\noindent \emph{Remark:} The key difficulty in the proof of Lemma \ref{lm:2ndmoment} is to obtain a strong concentration bound on the overlap $(\sigma, \tau)\mapsto  \langle \sigma, \tau\rangle$ under the product measure $\pf^{\otimes 2}$, restricted to $S_{\eps,k}\times S_{\eps,k} \subset \{-1, 1\}^N\times \{-1,1\}^N$. In our proof, the condition $\beta^2(1-q)\leq 1$ emerges due to the restrictions on one of the spin variables, say, $\tau \in S_{\eps,k}$. In principle, one may be able to find the optimal temperature condition by taking into account the restrictions on the other spin variable as well using standard large deviation variational estimates, but an exact solution seems difficult. We hope to get back to this point in future work.


\section{Proof of Theorem \ref{thm:main}}\label{sec:proofmain}

In this section we prove Theorem \ref{thm:main}, based on Lemmas \ref{lm:1stmoment} and \ref{lm:2ndmoment}. Before we start, let us first re-center the Hamiltonian $H_N$ appropriately, as outlined in the introduction. 

Using the notation of Section \ref{sec:mi}, we have that
		\[\begin{split}
		 \frac{H_N(\sigma)}N &= \frac\beta2 \langle \sigma, \bar{\tbf{g}} \sigma\rangle + \langle h, \sigma\rangle  = \frac\beta2 \langle \sigma, \bar{\tbf{g}}^{(k+1)} \sigma\rangle + \frac\beta2 \sum_{s=1}^k \langle \sigma, \bar\rho^{(s)}\sigma\rangle +\langle h,\sigma\rangle.
		\end{split}\]
In contrast to \cite{Bolt2}, instead of centering the spins $\sigma$ around $\tbf{m}^{(k+1)}$, we center the spins in $\langle \sigma, \bar\rho^{(s)}\sigma\rangle$ around $\gamma_s \phi^{(s)}$ in order to produce the right cavity field $\tbf{h}^{(k+1)}$. Notice that the remaining term $\langle \sigma, \bar{\tbf{g}}^{(k+1)} \sigma\rangle$ contains automatically centered spins around $\sum_{s=1}^k\gamma_s \phi^{(s)}$ (which approximately equals $\tbf{m}^{(k+1)}$), as $\bar{\tbf{g}}^{(k+1)} \phi^{(s)} = 0$ for $s<k+1$. We thus write
		\[\begin{split}
		\langle \sigma, \bar\rho^{(s)}\sigma\rangle &=2\gamma_s\langle \sigma , \bar\rho^{(s)}\phi^{(s)}\rangle+ \langle \sigma-\gamma_s \phi^{(s)}, \bar\rho^{(s)}(\sigma-\gamma_s \phi^{(s)})\rangle-\gamma_s^2\langle   \phi^{(s)}, \bar \rho^{(s) } \phi^{(s)} \rangle\\
		& = 2\gamma_s\langle   \sigma ,  \zeta^{(s)}\rangle+\langle \wh\sigma^{(s)} , \bar\rho^{(s)} \wh \sigma^{(s)}\rangle-\gamma_s^2\langle   \phi^{(s)}, \zeta^{(s)}\rangle,
		\end{split}\]
which follows from $ \bar \rho^{(s) } \phi^{(s)}  = \zeta^{(s)} $ and where we set $ \wh \sigma^{(s)}=\sigma-\gamma_s \phi^{(s)} $. Hence 
		\[\begin{split}
		 \frac{H_N(\sigma)}N & = \frac\beta2 \langle \sigma, \bar{\tbf{g}}^{(k+1)} \sigma\rangle  +\langle \tbf{h}^{(k+1)},\sigma\rangle\\
		 &\hspace{0.5cm} + \frac\beta2\sum_{s=1}^k \langle \wh\sigma^{(s)} , \bar\rho^{(s)} \wh \sigma^{(s)}\rangle- \frac\beta2\sum_{s=1}^k\gamma_s^2\langle   \phi^{(s)}, \zeta^{(s)}\rangle +\beta\Big( \gamma_k - \sqrt{q-\Gamma_{k-1}^2}\Big)\langle \sigma, \zeta^{(k)}\rangle.
		 \end{split}\]
Since an exact evaluation of the free energy seems rather involved, let us notice here that for configurations $\sigma \in S_{\eps,k}$ as defined in \eqref{eq:Sepsk} instead, we have approximately $  H_N(\sigma)/N  \approx \frac\beta2 \langle \sigma, \bar{\tbf{g}}^{(k+1)} \sigma\rangle  +\langle \tbf{h}^{(k+1)},\sigma\rangle$. Indeed, we find that
		\[\begin{split}
		\langle \wh\sigma^{(s)} , \bar\rho^{(s)} \wh \sigma^{(s)}\rangle = 2 \langle \sigma - \gamma_s\phi^{(s)}, \zeta^{(s)}\rangle ( \langle \phi^{(s)},\sigma\rangle -\gamma_s) - \langle \phi^{(s)},\zeta^{(s)}\rangle ( \langle \phi^{(s)},\sigma\rangle -\gamma_s)^2 
		\end{split}\]
with $ \gamma_s \simeq \langle \tbf{m}^{(k+1)}, \phi^{(s)}\rangle$. Similarly, recall that $ \langle   \phi^{(s)}, \zeta^{(s)}\rangle \sim \mathcal{N}(0, 2/N)$ for each $s$ and 
		$$\Big|\Big(\gamma_k-\sqrt{q-\Gamma_k^2}\Big) \sup_{\sigma\in \{-1,1\}^N} \langle \sigma, \zeta^{(k)}\rangle\Big| \leq \Big(\gamma_k+ \sqrt{q-\Gamma_k^2}\Big)\|\zeta^{(k)} \|$$
with $\lim_{k\to\infty} |q-\Gamma_k^2|=\lim_{k\to\infty}\gamma_k=0$ under \eqref{eq:ATline}, by Lemmas \ref{lm:seqlemma} and \ref{lm:aux}.

Thus, we obtain the simple lower bound
		\[\begin{split}
		\frac1N \log Z_N &=   \frac1N\log \sum_{\sigma\in\{-1,1\}^N} e^{H_N(\sigma)} \\
		& = \log 2 +\frac1N \sum_{i=1}^N\log\cosh(h^{(k+1)}_i)+ \frac1N\log \sum_{\sigma\in\{-1,1\}^N} \pf(\sigma) e^{H_N(\sigma)-N\langle \tbf{h}^{(k+1)},\sigma\rangle } \\
		 &\geq \log 2+\frac1N \sum_{i=1}^N\log\cosh(h^{(k+1)}_i)+\frac1N\log \sum_{\sigma\in S_{\eps,k}} \pf(\sigma) e^{H_N(\sigma)-N\langle \tbf{h}^{(k+1)},\sigma\rangle } \\
		\end{split}\]
so that 
		\begin{equation}\label{eq:ZNlowerbnd}
		\begin{split}
		\frac1N \log Z_N &\geq \log 2+\frac1N \sum_{i=1}^N\log\cosh(h^{(k+1)}_i)+\frac1N \log Z_N^{(k+1)}(S_{\eps,k}) \\
		&\hspace{0.4cm}- C \frac{\beta\eps}k \sum_{s=1}^k  \| \zeta^{(s)}\| - C\beta \sum_{s=1}^k  \| \zeta^{(s)}\|\big| \langle \phi^{(s)},\tbf{m}^{(k+1)}\rangle -\gamma_s\big|  \\
		&\hspace{0.4cm} - \frac\beta{2}\sum_{s=1}^k\gamma_s^2\langle   \phi^{(s)}, \zeta^{(s)}\rangle -\beta\gamma_k \| \zeta^{(k)} \|- \beta\sqrt{q-\Gamma_k^2} \| \zeta^{(k)} \|. 
		\end{split}
		\end{equation} 
We have now all necessary preparations for the proof of Theorem \ref{thm:main}.
\begin{proof}[Proof of Theorem \ref{thm:main}]
Up to minor modifications, we follow \cite[Section 3]{Bolt2} and we also abbreviate $RS(\beta,h)= \log 2 + E \log\cosh(\beta\sqrt q Z +h)+ \beta^2(1-q)^2/4.$ 

By the Paley-Zygmund inequality, we have that
		\[ \bP_k \Big(  Z_N^{(k+1)}(S_{\eps,k})\geq  \bE_kZ_N^{(k+1)}(S_{\eps,k})/2  \Big) \geq \frac{\big(\bE_k ( Z_N^{(k+1)}(S_{\eps,k})\big)^2}{4 \bE_k \big[ (Z_N^{(k+1)}(S_{\eps,k}))^2\big] }\]
Given $\delta_1 > 0$, Lemmas \ref{lm:1stmoment} and \ref{lm:2ndmoment} imply
		\[ \bP\bigg(  \frac2N \log  \bE_k \, Z_N^{(k+1)}(S_{\eps,k})  \geq  \frac1N \log 4\,\bE_k \big( Z_N^2(S_{\eps,k})\big)  - \delta_1 \bigg) \geq \frac12 \]
if we choose $\eps>0$ sufficiently small, $k$ sufficiently large and $N \geq N_1(\eps, k)\in \mathbb{N}$ sufficiently large. This also implies that
		\[\begin{split} 
		\bP\bigg( \bP_k \bigg( \frac1N\log Z_N^{(k+1)}(S_{\eps,k})\geq  \frac1N\log \bE_kZ_N^{(k+1)}(S_{\eps,k}) -\frac{\log 2}N \bigg) \geq e^{-\delta_1 N}  \bigg) \geq \frac12.
		\end{split}\]
On the other hand, applying Lemma \ref{lm:aux} \emph{3)}, we note that 
		\[ \lim_{N\to\infty} \bP \bigg( \frac1k \sum_{s=1}^k  \| \zeta^{(s)}\| > 2 \bigg) \leq \lim_{N\to\infty} \sum_{s=1}^k\bP \Big(   \| \zeta^{(s)}\| > 2 \Big)=0  \]
and that
		\[\lim_{N\to\infty} \bP \Big(\| \zeta^{(k+1)}\| >2 \Big) =0.  \]
Moreover, Lemma \ref{lm:aux} and the fact that $ \langle \phi^{(s)},\tbf{m}^{(k+1)}\rangle \simeq \gamma_s$ by Proposition \ref{prop:mi} imply 
		\[\lim_{N\to\infty } \bigg( C\beta \sum_{s=1}^k  \| \zeta^{(s)}\|\big| \langle \phi^{(s)},\tbf{m}^{(k+1)}\rangle -\gamma_s\big|  + \frac\beta{2}\sum_{s=1}^k\gamma_s^2\langle   \phi^{(s)}, \zeta^{(s)}\rangle \bigg) = 0\]
as well as
		\[ \lim_{N\to\infty} \frac1N \sum_{i=1}^N\log\cosh\big(h^{(k+1)}_i\big) = E \log\cosh(\beta\sqrt{q}Z +h) \]
in probability. Now, the lower bound \eqref{eq:ZNlowerbnd} implies that
		\[\begin{split} &\bP\bigg( \bP_k \bigg( \frac1N\log Z_N \geq \log 2 +  \frac1N\sum_{i=1}^N\log \cosh(\beta\sqrt{q}Z+h) \\
		&\hspace{3cm} + \frac1N\log \bE_kZ_N^{(k+1)}(S_{\eps,k}) -\frac{\log 2}N -\cE_{\eps,k} \bigg) \geq e^{-\delta_1 N}  \bigg) \geq \frac12, 
		\end{split}\] 
where we defined the error $\cE_{\eps,k}$ by
		\[\begin{split}
		\cE_{\eps,k}  &=    \frac{C\beta\eps}k \sum_{s=1}^k  \| \zeta^{(s)}\| + C\beta \sum_{s=1}^k  \| \zeta^{(s)}\|\big| \langle \phi^{(s)},\tbf{m}^{(k+1)}\rangle -\gamma_s\big|  \\
		&\hspace{0.4cm} + \frac\beta{2}\sum_{s=1}^k\gamma_s^2\langle   \phi^{(s)}, \zeta^{(s)}\rangle +\beta\gamma_k \| \zeta^{(k)} \| +  \beta\sqrt{q-\Gamma_k^2} \| \zeta^{(k)} \|.
		\end{split}\]	
Given $\delta_2>0$, we may choose $\eps>0$ sufficiently small, $k\in\mathbb{N}$ sufficiently large and $N\geq N_2(k,\eps) \in \mathbb{N} $ sufficiently large such that
		\[ \bP\bigg(| \cE_{\eps,k} | \leq  \frac{\delta_2}4, \,  \bigg| \frac1N \sum_{i=1}^N\log\cosh(h^{(k+1)}_i) - E \log\cosh(\beta\sqrt{q}Z +h)\bigg|\leq \frac{\delta_2}4\bigg) \geq \frac78 \]
and, by Lemma \ref{lm:1stmoment}, also such that
		\[ \bP \bigg( \frac1N \log \bE_k Z_N^{(k+1)}(S_{\eps,k}) \geq \frac{\beta^2}4(1-q)^2 - \frac{\delta_2}4\bigg) \geq \frac78. \]	
Combining the above observations, we find that
		\[ \bP\bigg( \bP_k \bigg( \frac1N\log Z_N \geq  RS(\beta,h)    - \delta_2 \bigg) \geq e^{-\delta_1 N}  \bigg) \geq \frac14 
		\] 	
for all $N\geq \max (N_1(\eps,k), N_2(\eps,k), 4\log2/\delta_2)$. This implies 
		\[ \bP\bigg(  \frac1N\log Z_N \geq  RS(\beta,h)    - \delta_2 \bigg)  \geq \frac14e^{-\delta_1 N}.
		\] 
By Gaussian concentration of the free energy, i.e.
		\[ \bP\bigg( \bigg| \frac1N\log Z_N - \frac1N \bE \log Z_N\bigg| \leq \delta_3 \bigg) \geq 1 - e^{ -N \delta_3^2 /\beta^2},\]
we may choose $ \delta_1< \delta_3^2/(2\beta^2) $ to conclude for large enough $N$ that 
		\[ \bE\frac1N  \log Z_N \geq RS(\beta,h) - \delta_2 -\delta_3. \]
Since $\delta_2, \delta_3>0$ were arbitrary, this shows that the right hand side in \eqref{eq:RS} is a lower bound to $\lim_{N\to\infty} \mathbb{E}\frac1N \log Z_N$ and by \emph{Remark} 2), this proves Theorem \ref{thm:main}.
\end{proof}

\noindent\textbf{Acknowledgements.} The work of H.-T. Y. is partially supported by NSF grant DMS-1855509 and a Simons Investigator award. The work of C.B. was partly funded by the Deutsche Forschungsgemeinschaft (DFG, German Research Foundation) under Germany's Excellence Strategy – GZ 2047/1, Projekt-ID 390685813.

\vspace{0.8cm}
\noindent {\Large \textbf{Data Availability}}
\vspace{0.5cm}

\noindent Data sharing is not applicable to this article as no new data were created or analyzed in this study.

\vspace{0.5cm}
\noindent {\Large \textbf{Conflict of Interests}}
\vspace{0.5cm}

\noindent The authors have no conflict of interest to declare that are relevant to the content of this article.





\begin{thebibliography}{55}

\bibitem{ABSY}
A. Adhikari, C. Brennecke, P. von Soosten, H.-T. Yau. Dynamical Approach to the TAP Equations for the Sherrington-Kirkpatrick Model. \emph{J. Stat. Phys.} \textbf{183}, 35 (2021). https://doi.org/10.1007/s10955-021-02773-7

\bibitem{ALR}
M. Aizenman, J. L. Lebowitz, D. Ruelle. Some rigorous results on the Sherrington-Kirkpatrick spin glass model. \emph{Comm. Math. Phys.} \textbf{11} (1987), pp. 3-20.

\bibitem{AT}
J. R. L. de Almeida, D. J. Thouless. Stability of the Sherrington--Kirkpatrick solution of a spin glass model. \emph{J. Phys. A: Math. Gen.} \textbf{11} (1978), pp. 983-990.

\bibitem{BelKis}
D. Belius, N. Kistler. The TAP-Plefka Variational Principle for the Spherical SK Model. \emph{Comm. Math. Phys.} \textbf{367} (2019), pp. 991-1017.

\bibitem{Bolt1}
E. Bolthausen. An Iterative Construction of Solutions of the TAP Equations for the Sherrington–Kirkpatrick Model. \emph{Comm. Math. Phys.} \textbf{325} (2014), pp. 333-366.

\bibitem{Bolt2}
E. Bolthausen. A Morita Type Proof of the replica-symmetric Formula for SK. In: \emph{Statistical Mechanics of Classical and Disordered Systems} (2018), Springer Proceedings in Mathematics \& Statistics, pp. 63-93. arXiv:1809.07972.

\bibitem{BoBo}
E. Bolthausen, A. Bovier. Spin Glasses. Lecture Notes in Mathematics, \textbf{1900} (2007), \emph{Springer Verlag Berlin--Heidelberg}.

\bibitem{Cha}
S. Chatterjee. Spin glasses and Stein’s method. \emph{Probab. Theory Relat. Fields} \textbf{148} (2010), pp. 567-600.

\bibitem{Ch}
W.-K. Chen. On the Almeida-Thouless transition line in the SK model with centered Gaussian external field. Preprint: arXiv:2103.04802.


\bibitem{Gue}
F. Guerra. Broken Replica Symmetry Bounds in the Mean Field Spin Glass Model. \emph{Comm. Math. Phys.} \textbf{233} (2003), pp. 1-12.

\bibitem{GueTo}
F. Guerra, F.L. Toninelli. The Thermodynamic Limit in Mean Field Spin Glass Models. \emph{Commun. Math. Phys.} \tbf{230} (2002), pp. 71-79.

\bibitem{AuJag}
A. Auffinger, A. Jagannath. Thouless-Anderson-Palmer equations for generic p-spin glasses. \emph{Ann. Probab.} \textbf{47} (2019), no. 4, pp. 2230--2256.

\bibitem{JagTob}
A. Jagannath, I. Tobasco. Some properties of the phase diagram for mixed p-spin glasses. \emph{Probab. Theory Relat. Fields} \textbf{167} (2017), pp. 615-672.

\bibitem{MPV}
M. M\'ezard, G. Parisi, M. A. Virasoro. Spin Glass Theory and Beyond. World Scientific Lecture Notes in Physics Vol.9 (1987), \emph{World Scientific, Singapore New Jersey Hong Kong}.

\bibitem{Pan1}
D. Panchenko. The Parisi ultrametricity conjecture. \emph{Ann. of Math.} \textbf{177} (2013), Issue 1, pp. 383-393.

\bibitem{Pan2}
D. Panchenko. The Sherrington-Kirkpatrick Model. Springer Monographs in Mathematics (2013), \emph{Springer Verlag, New York}.

\bibitem{Pan3}
D. Panchenko. The Parisi formula for mixed $p$-spin models. \emph{Ann. of Probab.} \textbf{42} (2014), no. 3, pp. 946-958.

\bibitem{Par1}
G. Parisi. Infinite number of order parameters for spin-glasses. \emph{Phys. Rev. Lett.} \textbf{43} (1979), pp. 1754-1756.

\bibitem{Par2}
G. Parisi. A sequence of approximate solutions to the S-K model for spin glasses. \emph{J. Phys. A} \textbf{13} (1980), pp. L-115.

\bibitem{Sh}

M. Shcherbina. On the Replica-Symmetric Solution for the Sherrington-Kirkpatrick Model. \emph{Helv. Phys. Acta} \textbf{70} (1997), pp. 838-853.

\bibitem{SK}
D. Sherrington, S. Kirkpatrick. Solvable model of a spin glass. \emph{Phys. Rev. Lett.} \textbf{35} (1975), pp. 1792-1796.


\bibitem{Tal}
M. Talagrand. The Parisi formula. \emph{Ann. of Math.} \textbf{163} (2006), no. 1, pp. 221-263.

\bibitem{Tal1}
M. Talagrand. Mean Field Models for Spin Glasses. Volume I: Basic Examples. A Series of Modern Surveys in Mathematics, \textbf{Vol. 54} (2011), \emph{Springer Verlag Berlin--Heidelberg}.

\bibitem{Tal2}
M. Talagrand. Mean Field Models for Spin Glasses. Volume II: Advanced Replica-Symmetry and Low Temperature. A Series of Modern Surveys in Mathematics, \textbf{Vol. 54} (2011), \emph{Springer Verlag Berlin--Heidelberg}.

\bibitem{TAP}
D. J. Thouless, P. W. Anderson, R. G. Palmer. Solution of 'solvable model in spin glasses'. \emph{Philos. Magazin} \textbf{35} (1977), pp. 593-601.

\end{thebibliography}
\end{document}